\documentclass[12pt,leqno]{article}

\usepackage{amsfonts,amsthm,amsmath,mathtools,amssymb,comment,enumitem}
\usepackage{aliascnt}
\usepackage{ascmac}

\usepackage{braket}
\usepackage{cleveref}

\theoremstyle{plain}
\newtheorem{thm}{Theorem}[section]
\newaliascnt{prop}{thm}   
\newtheorem{prop}[prop]{Proposition}
\aliascntresetthe{prop}
\newaliascnt{lem}{thm}   

\aliascntresetthe{lem}
\newaliascnt{cor}{thm}   

\aliascntresetthe{cor}
\theoremstyle{definition}
\newaliascnt{df}{thm}   
\newtheorem{df}[df]{Definition}
\aliascntresetthe{df}
\newaliascnt{rem}{thm}   
\newtheorem{rem}[rem]{Remark}
\aliascntresetthe{rem}
\newaliascnt{ex}{thm}   
\newtheorem{ex}[ex]{Example}
\aliascntresetthe{ex}
\newaliascnt{conj}{thm}   

\aliascntresetthe{conj}

\newcommand{\allone}{\mathbf{1}}

\newcommand{\Spec}{{\rm Spec}}

\newcommand{\RR}{\mathbb{R}}
\newcommand{\CC}{\mathbb{C}}
\newcommand{\FF}{\mathbb{F}}

\title{A New Approach to Code Smoothing Bounds}
\author{Tsuyoshi Miezaki\thanks{Faculty of Science and Engineering, Waseda University}
  \and
  Yusaku Nishimura\thanks{Graduate School of Fundamental Science and Engineering, Waseda University}
  \and
  Katsuyuki Takashima\thanks{Faculty of Education and Integrated Arts and Sciences, Waseda University} \thanks{Waseda Research Institute for Science and Engineering, Waseda University}}
\date{}

\begin{document}

\maketitle

\begin{abstract}
  Code smoothing is a phenomenon in which an error distribution makes a code statistically close to the uniform distribution over the ambient space. 
  This closeness is measured by the total variation distance. 
  Recently, Debris-Alazard et al.\ introduced a smoothing bound, which is an upper bound on this total variation distance. 
  Although the smoothing bound evaluates how the error distribution smooths a code, this bound applies only to linear codes. 
  In this paper, we generalize this bound to not only linear codes but also specific non-linear codes. 
  While the smoothing bound in previous work was obtained by Fourier analysis over finite abelian groups, we derive this bound using a graph-theoretic approach. 
  To derive the smoothing bound, we consider code smoothing as the mixing of random walks on a specific graph, and use the concept of equitable partitions, which is well-studied in graph theory.
  \end{abstract}

  \section{Introduction}
  
  \subsection{Background}
  The concept of smoothing was originally introduced to analyze the security of lattice-based cryptosystems, and recently, a similar approach has been proposed for codes. 
  An error distribution $f$ over $\mathbb{F}_2^n$ is said to smooth a code if the distribution obtained by adding error to the code is statistically close to the uniform distribution. 
  This closeness is evaluated by a sufficiently small total variation distance between the two distributions.
  Several researchers have attempted to use the code smoothing to establish a reduction from the decoding problem to the learning parity with noise (LPN) problem, that is, worst-case to average-case reductions \cite{BLV2019, DR2025, PB2025, YJ2021}.
  In this context, Debris-Alazard et al.\ \cite{DDRT2023} presented an upper bound on the total variation distance based on the Fourier analysis over locally compact abelian groups.
  This bound is referred to as a "smoothing bound," and it serves as one of the criteria for selecting the error distribution for a given code.
  For example, let $C$ be any code, and suppose that the error distribution $f_\omega$ is uniform over a Hamming sphere of radius $\omega$.
  Then, the smoothing bound establishes a condition on $\omega$ for $f_\omega$ to smooth $C$.
  

  Furthermore, it was pointed out in \cite{DDRT2023} that this total variation distance can also be viewed in the context of mixing in specific random walks. 
  It is well known that this distance is bounded by the eigenvalues of the corresponding transition matrices~\cite{LP2017, SC2004, W1997}. 
  Recently, it has been shown that by using a special subset of the state space, called a graphical design, as the initial state, a bound using even smaller eigenvalues can be obtained \cite{ST2025}.

  \subsection{Our Results}
  
  In this paper, we analyze the smoothing bound obtained in \cite{DDRT2023} from the perspective of random walks, and present Theorem~\ref{thm:main}, which generalizes this inequality for codes that are not necessarily linear.
  This method relies on equitable partitions, a concept closely related to the recently introduced graphical designs~\cite{S2020}.
  Although we aimed to improve the upper bound of the inequality by Debris-Alazard et al.\ using results from graphical designs, our approach did not yield an improvement over the existing bound.
  However, while the previous results in \cite{DDRT2023} rely on assumptions related to group structures, our bound (Theorem~\ref{thm:main}) only requires the existence of an equitable partition.
  Since equitable partitions can be obtained without relying on a group structure, Theorem~\ref{thm:main} can be applied to non-linear codes.
  Therefore, our work suggests that the concept of code smoothing can be extended to non-linear codes satisfying specific combinatorial properties.
  Furthermore, when the error distribution is either a Bernoulli or uniform over a Hamming sphere, codes satisfying the combinatorial conditions of Theorem~\ref{thm:main} can be interpreted in the context of graph theory.
  
  
  \subsection{Outline}
  This paper is organized as follows.
  In Section 2, we provide some terms related to codes and random walks, as well as the total variation distance.
  In Section 3, we consider the bound of the total variation distance using equitable partitions, and show that a result of Debris-Alazard et al.\ can be obtained in the case of codes.
  
  \section{Preliminaries}
  
  For any finite set $V$, we denote $\allone_V\in\mathbb{C}^{|V|}$ as the all-ones vector and $\allone_{X}\in\mathbb{C}^{|V|}$ as the characteristic vector of $X\subset V$.
  We also denote $u_V=\frac{\allone_V}{|V|}$ and $u_X=\frac{\allone_X}{|X|}$ as the uniform distributions over $V$ and $X$, respectively.
  When $V$ is clear, we omit $V$ and simply denote the all-ones vector and uniform distribution over $V$ as $\allone$ and $u$, respectively.
  We denote by $\langle \mathbf{x},\mathbf{y}\rangle$ the inner product of two complex vectors $\mathbf{x}=(x_1,\ldots,x_n)^{\top}$ and $\mathbf{y}=(y_1,\ldots,y_n)^{\top}$, that is, 
  \[
   \langle \mathbf{x},\mathbf{y}\rangle=\sum_{j=1}^n x_j\overline{y_j}.
  \]
  For any $\mathbf{x}=(x_1,\ldots,x_n)^\top\in \CC^n$, we denote by $\|\mathbf{x}\|_k$ the $L_k$-norm defined as follows:
  \[
   \|\mathbf{x}\|_k\coloneq\left(\sum_{i=1}^n |x_i|^k\right)^{\frac{1}{k}}.
  \]
  We also define $L_1(V)$ as the set of functions over $V$ whose $L_1$-norm is finite.
  For a square matrix $A$, we denote the multiset of its eigenvalues by $\Spec(A)$.

  \subsection{Total variation distance of a code}
  
  \begin{df}[Code, dual code]
   A subspace of $\mathbb{F}_q^n$ is called an $[n,k]_q$-linear code.
   The dual of an $[n,k]_q$-linear code $C$, denoted by $C^{\perp}$, is defined as follows:
   \[
   C^{\perp}\coloneq\{\mathbf{v}\in\mathbb{F}_q^n\mid \text{for all $\mathbf{c}\in C$, $\mathbf{v}\cdot\mathbf{c}=0$.}\}.
   \]
  \end{df}
  
  
  \begin{df}[Convolution]
   Let $G$ be a finite abelian group. 
   For any $f,h\in L_1(G)$, the convolution product is defined as follows:
   \[
   (f\ast h)(g)=\sum_{x\in G}f(g-x)h(x).
   \]
   We also define $f\ast^\ell h\coloneq f\ast h\ast\cdots\ast h$ as the $\ell$-fold convolution of $h$ with $f$, where $h$ appears $\ell$ times.
  \end{df}
  
  \begin{df}[Representation matrix]
   Let $G$ be a finite abelian group. 
   For any $f\in L_1(G)$, the representation matrix of $f$, denoted by $T_f$, is the square matrix indexed by $G$ whose $(c_1,c_2)$-entry is $f(c_1-c_2)$.
  \end{df}
  
  \begin{rem}
   For all $f,h\in L_1(G)$, we have $h\ast f=T_f h$, where functions on $G$ are identified with column vectors indexed by $G$.
  \end{rem}
  
  
  
  \subsection{Random walk}
  
  \begin{df}[Stochastic matrix]
   A matrix $A$ is called left stochastic if $A\allone=\allone$, and right stochastic if $\allone A=\allone$.
   In particular, when $A$ is both a left and right stochastic matrix, $A$ is called a doubly stochastic matrix.
  \end{df}
  
  \begin{df}[Random walk]
   Let $A$ be a left stochastic matrix indexed by a finite set $V$.
   A random walk over a finite set $V$ by a matrix $A$ is a Markov chain over $V$ with left transition matrix $A$.
  \end{df}
  
  We also define the total variation distance for a random walk.
  
  \begin{df}[Total variation distance of a random walk]
   Let $A$ be a left stochastic matrix indexed by a finite set $V$, and let $\mu_0$ be a probability distribution over $V$.
   For each nonnegative integer $\ell$, we write $\mu_\ell\coloneq A^\ell\mu_0$ for the probability distribution after $\ell$ steps of the random walk on $V$ with transition matrix $A$.
   The total variation distance of $\mu_0$ by a random walk with transition matrix $A$ after $\ell$ steps is defined as follows:
   \[
   \Delta_{\mbox{\small RW}}(u_V, \mu_\ell)\coloneq\frac{1}{2}\|u_V-A^\ell\mu_0\|_1.
   \]
   In particular, let $G$ be a finite abelian group, let $H$ be a subgroup of $G$, and let $f$ be a probability distribution over $G$.
   We also define the total variation distance of $H$ by $f$ after $\ell$ steps as follows:
   \[
      \Delta_{\mbox{\small Group}}(u_G, u_H\ast^\ell f)\coloneq\Delta_{\mbox{\small RW}}(u_G,T_f^\ell u_H).
   \]
  \end{df}
  
  
  \begin{rem}
      Since $u_H\ast^\ell f=T_f^\ell u_H$, we have
      \[\Delta_{\mbox{\small Group}}(u_G, u_H\ast^\ell f)=\frac{1}{2}\|u_G-u_H\ast^\ell f\|_1.\]
  \end{rem}

  \subsection{Previous work}
  
  In this subsection, we recall some fundamental facts from Fourier analysis and review a result of Debris-Alazard et al.~\cite{DDRT2023}.
  Hereafter, $G$ denotes a finite abelian group, and $\hat{G}$ denotes its character group.
  
  \begin{df}[Fourier transform]
   The Fourier transform of $f\in L_1(G)$, denoted by $\hat{f}$, is defined as follows:
   \[
   \hat{f}:\chi\in\hat{G}\mapsto\frac{1}{|G|}\sum_{g\in G}f(g)\overline{\chi(g)}.
   \]
  \end{df}
  
  The eigenvectors and eigenvalues of $T_f$ are determined by the Fourier transform of $f$.
  
  \begin{prop}\label{prop:spec}
   Let $f\in L_1(G)$.
   All characters $\chi\in \hat{G}$ are eigenvectors of $T_f$.
   Consequently, $T_f$ is normal and
  \[
   \Spec(T_f)=\{|G|\hat{f}(\chi) : \chi\in\hat{G}\}.
   \]
  \end{prop}

  \begin{proof}
    Let $\chi \in \hat{G}$ be a character of $G$, which we can view as a column vector indexed by the elements of $G$. 
    Then, for any $g_1 \in G$, the $g_1$-th entry of the vector $T_f \chi$ is given by:
    \begin{align*}
    (T_f \chi)_{g_1} &= \sum_{g_2 \in G} (T_f)_{g_1, g_2} \chi(g_2) \\
    &= \sum_{g_2 \in G} f(g_1 - g_2) \chi(g_2)\\
    &= \sum_{h \in G} f(h) \chi(g_1 - h)\\
    &= \chi(g_1)\sum_{h \in G} f(h) \overline{\chi(h)}\\
    &= |G|\hat{f}(\chi)\chi(g_1).
    \end{align*}
    Since this equality holds for all $g_1 \in G$, it follows that $T_f \chi = |G| \hat{f}(\chi) \chi$. 
    This shows that $\chi$ is an eigenvector of $T_f$ with corresponding eigenvalue $|G| \hat{f}(\chi)$.
    Since $|\hat{G}| = |G|$ and the characters are orthogonal, $T_f$ is normal and its spectrum is $\mathrm{Spec}(T_f) = \{|G| \hat{f}(\chi) \mid \chi \in \hat{G}\}$.
  \end{proof}

  For a more detailed discussion, see \cite[Section~8.2]{SC2004}.
  
  Next, we define the periodization of a function on a finite abelian group.
  
  \begin{df}[Periodization]
   Let $H$ be a subgroup of $G$.
   For all $f\in L_1(G)$, we define the periodization of $f$, denoted by $f^{\mid H}\in L_1(G/H)$, as follows:
   \[
   f^{\mid H}:(g+H)\in G/H\mapsto \frac{1}{|H|}\sum_{h\in H}f(g+h).
   \]
  \end{df}
  
  The Fourier transform of the periodization is given by the Poisson summation formula.
  
  \begin{thm}[The Poisson summation formula]\label{thm:poisson}
      Let $H$ be a subgroup of $G$.
      Let $(\hat{f})_{\mid \widehat{G\slash H}}$ be the restriction of $\hat{f}$ to $\widehat{G\slash H}$.
      Then,
      \[
      \widehat{f^{\mid H}}=(\hat{f})_{\mid \widehat{G\slash H}}.
      \]
  \end{thm}
     
  Using the Poisson summation formula, Debris-Alazard et al.\ obtained an upper bound on the total variation distance of a code.
  
  \begin{thm}[{\cite{DDRT2023}}]\label{thm:scupper}
      Let $C$ be an $[n,k]_2$-linear code, and let $f$ be a probability distribution over $\mathbb{F}_2^n$.
      Then
  \[
  \Delta_{\mbox{\small Group}}(u,u_C\ast^\ell f)\leq 
  2^{n-1}\sqrt{\sum_{x\in C^\perp\setminus\{0\}}|\hat{f}(x)|^{2\ell}}. 
  \]
  \end{thm}

  \section{The upper bound on the total variation distance using equitable partitions}\label{sec:tv}

  In this section, we give an upper bound on the total variation distance by a random walk when the initial distribution satisfies a specific condition.
  Additionally, we give an alternative proof of \Cref{thm:scupper} in terms of random walks.
  
  \subsection{The upper bound on the total variation distance of a random walk}\label{subsec:eq}
  
  In this subsection, we give two types of upper bounds on the total variation distance of a random walk (Theorem~\ref{thm:main}).
  The first type is similar to the result in \cite{ST2025}: the total variation distance is bounded by a subset of the eigenvalues of the transition matrix.
  The second type is tighter than the first and is obtained by imposing a specific condition on the eigenvectors of the transition matrix.
  
  We begin with the definition of an equitable partition.
  
  \begin{df}[Equitable partition]
   Let $A$ be a square matrix indexed by $V$, and let $P=\{V_1,\ldots,V_r\}$ be a partition of $V$.
   We define $A_{V_i,V_j}\in\mathbb{C}^{|V_i|\times|V_j|}$ as the submatrix of $A$ whose rows are indexed by $V_i$ and whose columns are indexed by $V_j$.
   For any $i,j\in\{1,\ldots,r\}$, if there exists a real number $q_{ij}$ such that $A_{V_i,V_j}\allone=q_{ij}\allone$, then $P$ is called an equitable partition of $A$.
   Additionally, the matrix $A^{\mid P}\in\mathbb{C}^{r\times r}$ whose $(i,j)$-entry is $q_{ij}$ is called the quotient matrix of $A$ with respect to $P$.
   We also define the characteristic matrix
   \[
   M_P\coloneq(\allone_{V_1},\ldots,\allone_{V_r}).
   \]
  \end{df}
  
  \begin{ex}\label{ex:hamming-partition}
   Let $f\in L_1(\FF_2^n)$ be defined by $f(x)=\frac{1}{n}$ if $x$ has exactly one non-zero component, and $f(x)=0$ otherwise.
   For $i\in\{0,\ldots,n\}$, let 
   \[V_i\coloneq\{x\in\FF_2^n\mid \text{$x$ has exactly $i$ non-zero components}\}.\]
   Then, $P=\{V_0,\ldots,V_n\}$ is an equitable partition of $T_f$.
   Note that $T_f$ is $\frac{1}{n}$ times the adjacency matrix of the Hamming graph on $\FF_2^n$.
  \end{ex}
  
  \begin{rem}
   When $A$ is the adjacency matrix of a graph, an equitable partition of $A$ corresponds to a perfect coloring of the graph.
   Therefore, \Cref{ex:hamming-partition} corresponds to a perfect coloring of the Hamming graph.
   In particular, for Hamming graphs, there exist several studies on perfect $2$-colorings, which are equitable partitions consisting of exactly two blocks~\cite{BKMTV2021,F2007,MV2020}.
  \end{rem}
  
  \begin{rem}\label{rem:cayley}
   Let $S$ be a nonempty subset of $G$, and let $f$ be defined by $f(x)=\frac{1}{|S|}$ if $x\in S$ and $f(x)=0$ otherwise.
   Then $T_f$ is $\frac{1}{|S|}$ times the adjacency matrix of a special graph associated with $G$ and $S$, known as a Cayley graph.
  \end{rem}
  
  From the definition of the equitable partition and quotient matrix, $P$ is an equitable partition of $A$ if and only if
  
  \begin{equation}
      AM_P=M_PA^{\mid P}.\label{eq:eq-partition}
  \end{equation}
  
  This equality implies some properties of the eigenvectors of $A$ and $A^{\mid P}$.
  
  \begin{prop}\label{prop:characteristic}
   Let $A$ be a square matrix indexed by $V$, and let $P=\{V_1,\ldots,V_r\}$ be an equitable partition of $A$.
   Let $\phi^{\mid P}$ be an eigenvector of $A^{\mid P}$ with an eigenvalue $\lambda$.
   Then, $M_P\phi^{\mid P}$ is an eigenvector of $A$ with the same eigenvalue $\lambda$.
  \end{prop}
  
  \begin{proof}
      From \Cref{eq:eq-partition},
      \[
      AM_P\phi^{\mid P}=M_PA^{\mid P}\phi^{\mid P}.
      \]
      Since $\phi^{\mid P}$ is an eigenvector of $A^{\mid P}$ with an eigenvalue $\lambda$,
      \[
      A^{\mid P}\phi^{\mid P}=\lambda\phi^{\mid P}.
      \]
      Therefore, $A(M_P\phi^{\mid P})=\lambda M_P\phi^{\mid P}$.
  \end{proof}
  
  \begin{prop}\label{prop:diagonalizable}
      Let $A$ be a square matrix, and let $P=\{V_1,\ldots,V_r\}$ be an equitable partition of $A$.
      If $A$ is diagonalizable, then $A^{\mid P}$ is diagonalizable.
  \end{prop}
  
  \begin{proof}
  From \Cref{eq:eq-partition}, for every polynomial $p$, $p(A)M_P=M_Pp(A^{\mid P})$.
  Especially, if $p(A)=0$, $M_Pp(A^{\mid P})=0$.
  Since $M_P^\top M_P$ is a diagonal matrix whose $(i,i)$-entry is $|V_i|$,
  $M_P^\top M_P$ is a regular matrix. 
  Therefore, $M_Pp(A^{\mid P})=0$ implies $p(A^{\mid P})=0$.
  Hence, the minimal polynomial of $A^{\mid P}$ divides that of $A$ and is therefore square-free.
  Thus, $A^{\mid P}$ is also diagonalizable.
  \end{proof}
  
   
  If a left stochastic matrix $A$ has an equitable partition $P$, then the total variation distance of the random walk with a specific initial distribution can be bounded by the spectrum of the quotient matrix $A^{\mid P}$.
  
  \begin{thm}\label{thm:main}
      Let $A$ be a normal and left stochastic matrix indexed by $V$, and let $P=\{V_1,\ldots,V_r\}$ be an equitable partition of $A$ with a quotient matrix $A^{\mid P}$.
      Then, for all $i\in\{1,\ldots,r\}$,
      \begin{equation}
      \Delta_{\mbox{\small RW}}(u_V,A^\ell u_{V_i})\leq\frac{1}{2}\sqrt{\frac{|V|}{|V_i|}\sum_{\lambda\in \Spec(A^{\mid P})\setminus\{1\}}|\lambda|^{2\ell}}.
      \label{eq:main-general}
      \end{equation}
      Additionally, for some $i\in\{1,\ldots,r\}$, if 
     \begin{equation}
      |\langle u_{V_i}, \phi\rangle|^2=\frac{1}{|V|}
      \label{eq:uniform-overlap}
      \end{equation}
      holds for all eigenvectors $\phi\in \mathrm{Im}(M_P)$ of $A$ with $\|\phi\|_2=1$, then
      \begin{equation}
      \Delta_{\mbox{\small RW}}(u_V,A^\ell u_{V_i})\leq\frac{1}{2}\sqrt{\sum_{\lambda\in \Spec(A^{\mid P})\setminus\{1\}}|\lambda|^{2\ell}}.
      \label{eq:main-strong}
      \end{equation}
  \end{thm}
     
     \begin{proof}
      Let $\psi_i\in\RR^r$ be a vector whose $i$-th entry is $1$ and others are $0$.
      From \Cref{prop:diagonalizable}, $A^{\mid P}$ is diagonalizable.
      Therefore, $\psi_i$ can be written as a linear combination of eigenvectors, that is, there exists $c_1,\ldots,c_r\in\mathbb{C}^*$ such that
      \[
      \psi_i=\sum_{j=1}^r c_j\phi_j^{\mid P},
      \]
      where $\phi_j^{\mid P}$ are the eigenvectors of $A^{\mid P}$ satisfying $\|M_P\phi_j^{\mid P}\|_2=1$.
      Therefore, 
      \[
      u_{V_i}=\frac{1}{|V_i|}M_P\psi_i=\frac{1}{|V_i|}\left(\sum_{j=1}^r c_j M_P\phi_j^{\mid P}\right).
      \]
      Since $A$ is a left stochastic matrix, $A^{\mid P}$ is also a left stochastic matrix, and in particular, $\allone$ is an eigenvector of $A^{\mid P}$.
      Hence, without loss of generality, we assume that $\phi_1^{\mid P}\coloneq\frac{1}{\sqrt{|V|}}\allone$.
      Let $\lambda_j$ be the eigenvalue corresponding to $\phi_j^{\mid P}$.
      From \Cref{prop:characteristic}, $M_P\phi_j^{\mid P}$ are eigenvectors of $A$.
      Additionally, because $A$ is normal,
      \[
      c_j=\langle M_P\psi_i, M_P\phi_j^{\mid P}\rangle,
      \]
      and especially $c_1=\frac{|V_i|}{\sqrt{|V|}}$.
      Thus,
      \begin{align*}
      A^\ell u_{V_i}=u_V +\sum_{j=2}^r \lambda_j^{\ell}\langle u_{V_i}, M_P\phi_j^{\mid P}\rangle M_P\phi_j^{\mid P}
      \end{align*}
      and
      \begin{align*}
      \left\|u_V-A^\ell u_{V_i}\right\|_1&=\left\|\sum_{j=2}^r \lambda_j^{\ell}\langle u_{V_i},M_P\phi_j^{\mid P}\rangle M_P\phi_j^{\mid P}\right\|_1\\
      &\leq\sqrt{|V|}\left\|\sum_{j=2}^r \lambda_j^{\ell}\langle u_{V_i},M_P\phi_j^{\mid P}\rangle M_P\phi_j^{\mid P}\right\|_2,
      \end{align*}
      where the last inequality follows from the Cauchy--Schwarz inequality.
      Since $\|M_P\phi_i^{\mid P}\|_2=1$ for all $i\in\{2,\ldots,j\}$ and all of them are orthogonal,
      \begin{align*}
      \left\|\sum_{j=2}^r \lambda_j^{\ell}\langle u_{V_i},M_P\phi_j^{\mid P}\rangle M_P\phi_j^{\mid P}\right\|_2^2
      &=\sum_{j=2}^r |\lambda_j|^{2\ell}|\langle u_{V_i},M_P\phi_j^{\mid P}\rangle|^2.
      \end{align*}
      Now, if $|\langle u_{V_i},M_P\phi_j^{\mid P}\rangle|^2=\frac{1}{|V|}$ holds, then
      \begin{align*}
      \left\|\sum_{j=2}^r \lambda_j^{\ell}\langle u_{V_i},M_P\phi_j^{\mid P}\rangle M_P\phi_j^{\mid P}\right\|_2^2&=\frac{1}{|V|}\sum_{\lambda\in \Spec(A^{\mid P})\setminus\{1\}} |\lambda_j|^{2\ell}.
      \end{align*}
      Therefore, 
     \begin{align*}
      \left\|u_V-A^\ell u_{V_i}\right\|_1&\leq\sqrt{|V|}\left\|\sum_{j=2}^r \lambda_j^{\ell}\langle u_{V_i},M_P\phi_j^{\mid P}\rangle M_P\phi_j^{\mid P}\right\|_2\\
      &=\sqrt{\sum_{\lambda\in \Spec(A^{\mid P})\setminus\{1\}} |\lambda_j|^{2\ell}}.
      \end{align*}
      Even if such a condition does not hold, because
      \[|u_{V_i}|^2=\frac{1}{|V_i|},\]
      using the Cauchy--Schwarz inequality yields
     \begin{align*}
      \sqrt{\sum_{j=2}^r |\lambda_j^{2\ell}|\langle u_{V_i},M_P\phi_j^{\mid P}\rangle^2}&\leq\sqrt{\sum_{j=2}^r |\lambda_j|^{2\ell}\|u_{V_i}\|_2^2\|M_P\phi_j^{\mid P}\|_2^2}\\
      &=\sqrt{\frac{1}{|V_i|}\sum_{\lambda\in \Spec(A^{\mid P})\setminus\{1\}} |\lambda_j|^{2\ell}}.
     \end{align*}
      Therefore,
      \[
      \left\|u_V-A^\ell u_{V_i}\right\|_1\leq\sqrt{\frac{|V|}{|V_i|}\sum_{\lambda\in \Spec(A^{\mid P})\setminus\{1\}} |\lambda_j|^{2\ell}}.
      \]
     \end{proof}
     
  
  
  \begin{ex}\label{ex:1}
   Let $f$ and $P=\{V_0,\ldots,V_n\}$ be as in \Cref{ex:hamming-partition}.
   Then we can analyze $\Delta_{\mbox{\small RW}}(u,T_f^\ell u_{V_i})$ using \Cref{thm:main}.
   Note that since $V_i$ is not a linear subspace of $\FF_2^n$, we cannot apply \Cref{thm:scupper}.
  \end{ex}
  
\begin{rem}
    Let $f$ and $P=\{V_0,\ldots,V_n\}$ be as in \Cref{ex:hamming-partition}.
    Let $p\in[0,\frac{1}{2}]$, and let $f_{\mathrm{ber},p}$ denote the Bernoulli distribution on $\FF_2^n$ given by
    \[f_{\mathrm{ber},p}(x)=p^{|x|}(1-p)^{n-|x|},\]
    where $|x|$ denotes the number of $1$'s in $x$.
    It is known that any equitable partition of $T_f$ is also an equitable partition of $T_{f_{\mathrm{ber},p}}$.
    Therefore, we can apply \Cref{thm:main} to $\Delta_{\mbox{\small RW}}(u,T_{f_{\mathrm{ber},p}}^\ell u_{V_i})$.
\end{rem}

  \Cref{thm:main} states that if all the eigenvectors of the quotient matrix satisfy \Cref{eq:uniform-overlap}, then the upper bound on the $L_1$-norm becomes \Cref{eq:main-strong}.
   When considering $\Delta_{\mbox{\small Group}}(u_G,u_H\ast^\ell f)$, we can always apply \Cref{eq:main-strong}, which exactly coincides with \Cref{thm:scupper}.
  
  \subsection{The upper bound on the total variation distance of a code}\label{subsec:code}
  
   In this subsection, we prove one of the generalizations of \Cref{thm:scupper} using \Cref{thm:main}.
   Hereinafter, let $G$ be a finite abelian group.
  
   First, we show that for any $f\in L_1(G)$, the partition of $G$ into cosets of any subgroup is an equitable partition of $T_f$.
  
   \begin{prop}\label{prop:quo}
   Let $H$ be a subgroup of $G$, and 
  \[
   P_H=\{g_0+H,\ldots,g_{m-1}+H\}
   \]
   be a partition of $G$ into cosets of $H$.
   Then, for any $f \in L_1(G)$, $P_H$ is an equitable partition of $T_f$, and the $(g_i+H,g_j+H)$-entry of the quotient matrix is given by $\sum_{h'\in H}f(g_i-g_j+h')$.
   \end{prop}
  
  \begin{proof}
   Let $(T_f)_{g_i+H,g_j+H}$ be the submatrix of $T_f$ corresponding to rows indexed by $g_i+H$ and columns indexed by $g_j+H$,
   and let $x_h$ be the $h$-entry of $(T_f)_{g_i+H,g_j+H}\allone$.
   Then,
   \begin{align*}
   x_h&=\sum_{h'\in H}f(g_i+h-(g_j+h'))\\
   &=\sum_{h'\in H}f(g_i-g_j+h').
   \end{align*}
   Therefore,
   \[
   (T_f)_{g_i+H,g_j+H}\allone=\sum_{h'\in H}f(g_i-g_j+h')\allone.
   \]
  \end{proof}
  
  In the following, we denote the quotient matrix of $T_f$ with respect to $P_H$ by $T_f^{\mid H}$.
  Note that from the proof of \Cref{prop:quo}, the $(g_i+H,g_j+H)$-entry of $T_f^{\mid H}$ is $\sum_{h\in H}f(g_i-g_j+h)$.

  \begin{prop}\label{prop:period}
   For any $f\in L_1(G)$,
   \[
   |H|T_{f^{\mid H}}=T_f^{\mid H}.
   \]
  \end{prop}
  \begin{proof}
   From the definition of periodization, the $(g_i+H,g_j+H)$-entry of $T_{f^{\mid H}}$ is $\frac{1}{|H|}\sum_{h\in H}f(g_i-g_j+h)$.
   On the other hand, from \Cref{prop:quo}, the $(g_i+H,g_j+H)$-entry of $T_f^{\mid H}$ is $\sum_{h\in H}f(g_i-g_j+h)$.
  \end{proof}
  
  Therefore, the eigenvalues of $T_f^{\mid H}$ can be obtained from the eigenvalues of $T_{f^{\mid H}}$.
  From \Cref{prop:spec} and \Cref{thm:poisson}, the eigenvalues of $T_{f^{\mid H}}$ are
  \[\{|G/H|\hat{f}(\chi)\mid \chi\in \widehat{G\slash H}\}.\]
  In particular, the eigenvalue corresponding to $\allone$ is $\hat{f}(\chi_0)$, where $\chi_0$ is the identity of $\hat{G}$.
  Therefore, from \Cref{thm:main}, we obtain the following theorem.
  
  \begin{thm}\label{thm:group}
   Let $f$ be a probability distribution over $G$, let $H$ be a subgroup of $G$, and define $g+H\coloneqq\{g+h\in G\mid h\in H\}$, where $g$ is any element of $G$.
  Then, 
  \[
   \Delta_{\mbox{\small Group}}(u_G,u_{g+H}\ast^\ell f)\leq \frac{|G|^\ell}{2}\sqrt{\sum_{\chi\in \widehat{G\slash H}\setminus\{\chi_0\}}(\hat{f}(\chi))^{2\ell}}.
   \]
  \end{thm}
  
  \begin{proof}
   Let $[G:H]=r$.
   From \Cref{prop:quo}, a partition 
  \[
   P_H=\{H,g_1+H,\ldots,g_{r-1}+H\}
   \]
   is an equitable partition of $T_f$.
   Additionally, from \Cref{prop:period}, the eigenvectors of $T_f^{\mid H}$ are equal to those of $T_{f^{\mid H}}$,
   and from \Cref{prop:spec}, these eigenvectors are characters of $G\slash H$ multiplied by a scalar.
   In particular, since for any $g\in G$, $|\phi(g+H)|=1$ holds for any character $\phi\in\widehat{G\slash H}$, 
   $\|M_{P_H}\frac{\phi}{\sqrt{|G|}}\|_2=1$ and 
  \[
   \left|\left\langle u_{g+H}, M_{P_H}\frac{\phi}{\sqrt{|G|}}\right\rangle\right|^2=\frac{1}{|G|}.
   \]
   Therefore, we can use \Cref{eq:main-strong}:
   \begin{align}
   2\Delta_{\mbox{\small Group}}(u_G,u_{g+H}\ast^\ell f)&=2\Delta_{\mbox{\small RW}}(u_G,T_f^\ell u_{g+H})\nonumber\\
   &\leq \sqrt{\sum_{\lambda\in \Spec(T_f^{\mid H})\setminus\{1\}}|\lambda|^{2\ell}}.\label{eq:coset-mixing}
   \end{align}
   Additionally, from \Cref{prop:period}, 
  \begin{equation}
   T_f^{\mid H}=|H|T_{f^{\mid H}}\label{eq:periodization-matrix}
   \end{equation}
   and from \Cref{prop:spec},
   \[
   \Spec(T_{f^{\mid H}})=\left\{|G\slash H|\widehat{f^{\mid H}}(\chi)\mid \chi\in\widehat{G\slash H}\right\}.
   \]
   Using \Cref{thm:poisson},
   \[
   \Spec(T_{f^{\mid H}})=\left\{|G\slash H|\hat{f}(\chi)\mid \chi\in\widehat{G\slash H}\right\}.
   \]
   Therefore, from \Cref{eq:periodization-matrix},
   \[
   \Spec(T_f^{\mid H})=\left\{|G|\hat{f}(\chi)\mid \chi\in\widehat{G\slash H}\right\}.
   \]
   From \Cref{eq:coset-mixing},
   \begin{align*}
   2\Delta_{\mbox{\small RW}}(u,A^\ell u_{g+H})&\leq \sqrt{\sum_{\chi\in\widehat{G\slash H}\setminus\{\chi_0\}}||G|\hat{f}(\chi)|^{2\ell}}\\
   &=|G|^\ell\sqrt{\sum_{\chi\in\widehat{G\slash H}\setminus\{\chi_0\}}|\hat{f}(\chi)|^{2\ell}}.
   \end{align*}
  \end{proof}
  
  While \Cref{thm:scupper} gives an upper bound only for the case when the initial distribution is over a subgroup, \Cref{thm:group} gives an upper bound even if the initial distribution is over one of the cosets.
  \Cref{thm:group} also implies that \Cref{eq:main-strong} is as strong as the bound in \Cref{thm:scupper}.
  Therefore, if there exists a left stochastic matrix $A$, its equitable partition $P$, and a block $V_i\in P$ such that 
  these satisfy \Cref{eq:uniform-overlap}, then the total variation distance of the random walk with a initial distribution $u_{V_i}$ has an upper bound as strong as \Cref{thm:scupper}.
  
  
  
  \section{Conclusion}
  

  In this paper, we extend the smoothing bound to codes that are not necessarily linear. 
  For linear codes, our smoothing bound corresponds to the one established in previous work~\cite{DDRT2023}.
  Therefore, our bound implies that there exist several non-linear codes that can estimate the smoothness as effectively as linear codes. 
  Additionally, when $f$ is of the special form in \Cref{rem:cayley}, it suffices to find a perfect coloring of the corresponding graph.
  Future work is to construct non-linear codes whose smoothing bounds are as strong as those of linear codes.

  \section*{Acknowledgements}
  The authors would like to thank Professor Akihiro Munemasa for pointing out the diagonalizability of the quotient matrix.
  This work is supported by JSPS Grant-in-Aid for Scientific Research(C) JP22K11912, 
  JST K Program Grant Number JPMJKP24U2, 
  MEXT Quantum Leap Flagship Program (MEXT Q-LEAP) JPMXS0120319794,
  JSPS KAKENHI (25K06927),
  the Yoshida Scholarship Foundation through the Doctor 21 program, 
  and a Waseda Research Institute for Science and Engineering Grant-in-Aid for Young Scientists (Early Bird).

\end{document}